\documentclass{article}
\usepackage{spconf,amsmath,graphicx}

\usepackage{amsmath,amsfonts}
\usepackage{amssymb}
\usepackage{algorithm}
\usepackage{array}
\usepackage[caption=false,font=normalsize,labelfont=sf,textfont=sf]{subfig}
\usepackage{textcomp}
\usepackage{stfloats}
\usepackage{url}
\usepackage{verbatim}
\usepackage{graphicx}
\usepackage{cite}

\usepackage[acronym]{glossaries}

\newacronym{2d}{2D}{two-dimensional}
\newacronym{3d}{3D}{three-dimensional}
\newacronym{3gpp}{3GPP}{3rd generation partnership project}
\newacronym{4g}{4G}{fourth-generation}
\newacronym{5g}{5G}{fifth-generation}
\newacronym{6g}{6G}{sixth-generation}
\newacronym{abp}{ABP}{authentication by personalisation}
\newacronym{aclr}{ACLR}{adjacent channel leakage ratio}
\newacronym{adc}{ADC}{Analog-to-digital converter}
\newacronym{adr}{ADR}{adaptive data rate}
\newacronym{ag}{AG}{array gain}
\newacronym{ai}{AI}{artificial intelligence}
\newacronym{amam}{AM/AM}{amplitude modulation to amplitude modulation}
\newacronym{amp}{AMP}{approximate message passing}
\newacronym{ampm}{AM/PM}{amplitude modulation to phase modulation}
\newacronym{aoa}{AOA}{angle-of-arrival}
\newacronym{aod}{AOD}{angle-of-departure}
\newacronym{ap}{AP}{access point}
\newacronym{apu}{APU}{access point unit}
\newacronym{ar}{AR}{augmented reality}
\newacronym{arp}{ARP}{Antenna Reference Point}
\newacronym{asic}{ASIC}{application specific integrated circuit}
\newacronym{ask}{ASK}{amplitude-shift keying}
\newacronym{auv}{AUV}{autonomous underwater vehicle}
\newacronym{awgn}{AWGN}{additive white Gaussian noise}
\newacronym{baw}{BAW}{bulk acoustic wave}
\newacronym{bb}{BB}{base-band}
\newacronym{bf}{BF}{beamforming}
\newacronym{bldc}{BLDC}{brushless DC}
\newacronym{bom}{BOM}{bill of materials}
\newacronym{bs}{BS}{base station}
\newacronym{bw}{BW}{bandwidth}
\newacronym{cad}{CAD}{channel activity detection}
\newacronym{cars}{CARS}{calibration reference signal}
\newacronym{cbm}{CBM}{condition based maintenance}
\newacronym{cc}{CC}{constant current}
\newacronym{ccdf}{CCDF}{complementary cumulative distribution function}
\newacronym{ccnn}{CCNN}{circular convolutional neural network}
\newacronym{ccs}{CCS}{correlative channel sounder}
\newacronym{cdf}{CDF}{cumulative distribution function}
\newacronym{cdrx}{CDRX}{connected mode DRX}
\newacronym{ce}{CE}{coverage enhancement}
\newacronym{ced}{CED}{cumulative energy density}
\newacronym{cf}{CF}{cell-free}
\newacronym{cfmmimo}{CF-mMIMO}{cell-free massive MIMO}
\newacronym{cfo}{CFO}{carrier frequency offset}
\newacronym{cir}{CIR}{channel impulse response}
\newacronym{cost}{COST}{commercial off-the-shelf}
\newacronym{cots}{COTS}{commercial off-the-shelf}
\newacronym{cp}{CP}{cyclic prefix}
\newacronym{cpt}{CPT}{capacitive power transfer}
\newacronym{cpu}{CPU}{central-processing unit}
\newacronym{cqi}{CQI}{channel quality indicator}
\newacronym{cr}{CR}{coding rate}
\newacronym{crlb}{CRLB}{Cram\'er-Rao lower bound}
\newacronym{crs}{CRS}{cell reference signal}
\newacronym{crc}{CRC}{cyclic redundancy check}
\newacronym{cs}{CS}{compressed sensing}
\newacronym{csi}{CSI}{channel state information}
\newacronym{csp}{CSP}{contact service point}
\newacronym{css}{CSS}{chirp spread spectrum}
\newacronym{cu}{CU}{central unit}
\newacronym{cv}{CV}{constant voltage}
\newacronym{dac}{DAC}{digital-to-analog converter}
\newacronym{daq}{DAQ}{data acquisition system}
\newacronym{dc}{DC}{direct current}
\newacronym{dcc}{DCC}{dynamic cooperation clustering}
\newacronym{de}{DE}{drain efficiency}
\newacronym{dl}{DL}{downlink}
\newacronym{dlc}{DLC}{Distributed Laser Charging}
\newacronym{dma}{DMA}{Direct Memory Access}
\newacronym{dmimo}{D-MIMO}{distributed MIMO}
\newacronym{dpd}{DPD}{digital pre-distortion}
\newacronym{dpdk}{DPDK}{Data Plane Development Kit}
\newacronym{drx}{DRX}{Discontinuous Reception Mode}
\newacronym{duc}{DUC}{digital up-converter}
\newacronym{e2e}{E2E}{end-to-end}
\newacronym{easa}{EASA}{European Union Aviation Safety Agency}
\newacronym{ecdf}{eCDF}{empirical cumulative distribution function}
\newacronym{ecsp}{ECSP}{edge computing service point}
\newacronym{edlc}{EDLC}{electrostatic double-layer capacitors}
\newacronym{edrx}{eDRX}{Extended Discontinuous Reception Mode}
\newacronym{ee}{EE}{energy efficiency}
\newacronym{egprs}{EGPRS}{Enhanced Data Rates for GSM Evolution}
\newacronym{eirp}{EIRP}{effective isotropic radiated power}
\newacronym{em}{EM}{electromagnetic}
\newacronym{embb}{eMBB}{enhanced Mobile Broadband}
\newacronym{en}{EN}{energy neutral}
\newacronym{eol}{EoL}{end of life}
\newacronym{epu}{EPU}{edge processing unit}
\newacronym{erp}{ERP}{effective radiated power}
\newacronym[plural=ESCs,firstplural=electronic speed controllers (ESCs)]{esc}{ESC}{electronic speed control}
\newacronym{etsi}{ETSI}{European Telecommunications Standards Institute}
\newacronym{evm}{EVM}{Error Vector Magnitude}
\newacronym{fa}{FA}{federation anchor}
\newacronym{fembb}{feMBB}{further enhanced mobile broadband}
\newacronym{fft}{FFT}{fast Fourier transform}
\newacronym{fh}{FH}{fronthaul}
\newacronym{fim}{FIM}{Fisher information matrix}
\newacronym{fom}{FoM}{figure of merit}
\newacronym{fpga}{FPGA}{field-programmable gate array}
\newacronym{gb}{GB}{grant-based}
\newacronym{gf}{GF}{grant-free}
\newacronym{gnb}{gNB}{Next Generation Node B}
\newacronym{gnn}{GNN}{graph neural network}
\newacronym{gnss}{GNSS}{global navigation satellite system}
\newacronym{gpclk}{GPCLK}{general purpose clock}
\newacronym{gprs}{GPRS}{General Packet Radio Services}
\newacronym{gps}{GPS}{Global Positioning System}
\newacronym{gpu}{GPU}{graphical processing unit}
\newacronym{gsm}{GSM}{Global System for Mobile Communications}
\newacronym{gwp}{GWP}{Global Warming Potential}
\newacronym{gscm}{GSCM}{geometry‐based stochastic model}
\newacronym{harq}{HARQ}{hybrid automatic repeat request}
\newacronym{hat}{HAT}{hardware attached on top}
\newacronym{hcs}{HCS}{human-centric services}
\newacronym{hpbm}{HPBM}{half power beam width}
\newacronym{HyMPRo}{HyMPRo}{Hybrid Multi-Path Routing algorithm}
\newacronym{ib}{IB}{in-band}
\newacronym{ibo}{IBO}{input back-off}
\newacronym{ic}{IC}{integrated circuit}
\newacronym{id}{ID}{information decoding}
\newacronym{if}{IF}{intermediate-frequency}
\newacronym{iid}{i.i.d.}{independently and identically distributed}
\newacronym{im}{IM}{intermodulation}
\newacronym{imu}{IMU}{inertial measurement unit}
\newacronym{io}{IO}{input/output}
\newacronym{iot}{IoT}{Internet of Things}
\newacronym{ipt}{IPT}{inductive power transfer}
\newacronym{ipy}{IPY}{Interventions per Year}
\newacronym{iq}{IQ}{in-phase and quadrature}
\newacronym{ir}{IR}{infrared}
\newacronym{ism}{ISM}{industrial, scientific and medical}
\newacronym{isp}{ISP}{internet service provider}
\newacronym{kpi}{KPI}{key performance indicator}
\newacronym{kvi}{KVI}{key value indicator}
\newacronym{lca}{LCA}{life cycle assessment}
\newacronym{lco}{LCO}{lithium cobalt oxide}
\newacronym{ldo}{LDO}{Low-dropout voltage regulator}
\newacronym{ldpc}{LDPC}{low-density parity-check}
\newacronym{lfp}{LFP}{lithium iron phosphate}
\newacronym{lic}{LIC}{lithium-ion capacitor}
\newacronym{lidar}{LiDAR}{light detection and ranging}
\newacronym{liion}{Li-ion}{lithium-ion}
\newacronym{lipo}{LiPo}{lithium polymer}
\newacronym{lis}{LIS}{large intelligent surface}
\newacronym{llh}{LLH}{log-likelihood}
\newacronym{lmmse}{LMMSE}{least minimum mean square error}
\newacronym{lmo}{LMO}{lithium ion manganese oxide}
\newacronym{lo}{LO}{local oscillator}
\newacronym{lora}{LoRa}{long range}
\newacronym{lorawan}{LoRaWAN}{long-range wide-area network}
\newacronym{los}{LoS}{line-of-sight}
\newacronym{lp}{LP}{linear programming}
\newacronym{lpt}{LPT}{laser power transfer}
\newacronym{lpwa}{LPWA}{Low Power Wide Area}
\newacronym{lpwan}{LPWAN}{low-power wide-area network}
\newacronym{lqi}{LQI}{link quality indicator}
\newacronym{lrelu}{LReLU}{leaky rectified linear unit}
\newacronym{lrt}{LRT}{likelihood-ratio test}
\newacronym{ls}{LS}{least squares}
\newacronym{lsa}{LSA}{large synthetic array}
\newacronym{lsf}{LSF}{large-scale fading}
\newacronym{lsfc}{LSFC}{large-scale fading component}
\newacronym{ltc}{LTC}{lithium thionyl chloride}
\newacronym{lte}{LTE}{Long Term Evolution}
\newacronym{lto}{LTO}{lithium titanate}
\newacronym{m2m}{M2M}{machine to machine}
\newacronym{mac}{MAC}{Medium Access Control}
\newacronym{mate}{MATE}{millimeter-wave MIMO testbed}
\newacronym{mcl}{MCL}{Maximum Coupling Loss}
\newacronym{mcs}{MCS}{modulation and coding scheme}
\newacronym{mcu}{MCU}{Microcontroller Unit}
\newacronym{mec}{MEC}{multi-access edge computing}
\newacronym{mems}{MEMS}{micro-electromechanical systems}
\newacronym{mf}{MF}{matched filter}
\newacronym{mimo}{MIMO}{multiple-input multiple-output}
\newacronym{miso}{MISO}{multiple-input single-output}
\newacronym{ml}{ML}{machine learning}
\newacronym{mlp}{MLP}{multilayer perceptron}
\newacronym{mmimo}{mMIMO}{massive MIMO}
\newacronym{mmse}{MMSE}{minimum mean square error}
\newacronym{mmtc}{mMTC}{massive machine-typed communication}
\newacronym{mmwave}{mmWave}{millimeter wave}
\newacronym{mpc}{MPC}{multipath component}
\newacronym{mppt}{MPPT}{maximum power point tracking}
\newacronym{mr}{MR}{maximum ratio}
\newacronym{mrc}{MRC}{maximum ratio combining}
\newacronym{mrc_em}{MRC}{maximum ratio combining}
\newacronym{mrc_EM}{MRC}{Magnetic Resonance Coupling}
\newacronym{mrt}{MRT}{maximum ratio transmission}
\newacronym{mse}{MSE}{mean square error}
\newacronym{mtc}{MTC}{Machine-Type Communication}
\newacronym{multi-rat}{Multi-RAT}{multiple radio access technology}
\newacronym{nb}{NB}{narrowband}
\newacronym{nbiot}{NB-IoT}{narrowband IoT}
\newacronym{nca}{NCA}{nickel cobalt aluminum}
\newacronym{nicd}{NiCd}{nikkel cadmium}
\newacronym{nimh}{NiMH}{nikkel metal hydride}
\newacronym{nlos}{NLoS}{non-line-of-sight}
\newacronym{nmc}{NMC}{nickel manganese cobalt}
\newacronym{nn}{NN}{neural network}
\newacronym{nnls}{NNLS}{non-negative least squares}
\newacronym{noma}{NOMA}{non-orthogonal multiple access}
\newacronym{nr}{NR}{New Radio}
\newacronym{ntp}{NTP}{network time protocol}
\newacronym{oai}{OAI}{OpenAirInterface} 
\newacronym{ofdm}{OFDM}{orthogonal frequency-division multiplexing}
\newacronym{ofdmim}{OFDM-IM}{OFDM with index modulation}
\newacronym{oob}{OOB}{out-of-band}
\newacronym{oran}{O-RAN}{open radio-access network}
\newacronym{os}{OS}{operating system}
\newacronym{ota}{OTA}{over-the-air}
\newacronym{otaa}{OTAA}{over-the-air authentication}
\newacronym{p2p}{P2P}{point-to-point}
\newacronym{pa}{PA}{power amplifier}
\newacronym{pae}{PAE}{power-added efficiency}
\newacronym{papr}{PAPR}{peak-to-average power ratio}
\newacronym{pc}{PC}{pilot count}
\newacronym{pcb}{PCB}{printed circuit board}
\newacronym{pcg}{PCG}{power consumption gain}
\newacronym{pd}{PD}{powered device}
\newacronym{pdcch}{PDCCH}{physical downlink control channel}
\newacronym{pdp}{PDP}{power delay profile}
\newacronym{pdsch}{PDSCH}{physical downlink shared channel}
\newacronym{per}{PER}{packet error rate}
\newacronym{pg}{PG}{path gain}
\newacronym{pgd}{PGD}{proximal gradient descent}
\newacronym{pl}{PL}{path loss}
\newacronym{pll}{PLL}{phase-locked loop}
\newacronym{poe}{PoE}{power-over-Ethernet}
\newacronym{pps}{1PPS}{pulse per second}
\newacronym{prbs}{PRBs}{Physical Resource Blocks}
\newacronym{ps}{PS}{Processing System}
\newacronym{psd}{PSD}{power spectral density}
\newacronym{pse}{PSE}{power sourcing equipment}
\newacronym{psm}{PSM}{power saving mode}
\newacronym{pss}{PSS}{primary synchronisation signal}
\newacronym{ptp}{PTP}{precision-time protocol}
\newacronym{ptrs}{PTRS}{Phase-Tracking Reference Signals}
\newacronym{ptw}{PTW}{paging time window}
\newacronym{pwm}{PWM}{pulse width modulation}
\newacronym{qam}{QAM}{quadrature amplitude modulation}
\newacronym{qos}{QoS}{quality-of-service}
\newacronym{quadriga}{QuaDRiGa}{QUAsi Deterministic RadIo channel GenerAtor}
\newacronym{ra}{RA}{Random Access}
\newacronym{ran}{RAN}{radio access network}
\newacronym{rar}{RAR}{Random Access Response}
\newacronym{rat}{RAT}{radio access technology}
\newacronym{rbs}{RBS}{radio base station}
\newacronym{re}{RE}{radio element}
\newacronym[]{relu}{ReLU}{rectified linear unit}
\newacronym{rf}{RF}{radio frequency}
\newacronym{rfeh}{RFEH}{radio frequency energy harvesting}
\newacronym{rfic}{RFIC}{radio-frequency integrated circuit}
\newacronym{rfid}{RFID}{radio frequency identification}
\newacronym{rfpt}{RFPT}{radio frequency power transfer}
\newacronym{rfsoc}{RFSoC}{Radio Frequency System-on-Chip}
\newacronym{ris}{RIS}{reflective intelligent surface}
\newacronym{rllmtc}{RLLMTC}{reliable low latency machine type communication}
\newacronym{rms}{RMS}{root-mean-square}
\newacronym{rmse}{RMSE}{root-mean-square error}
\newacronym{rmt}{RMT}{random matrix theory}
\newacronym{rof}{RoF}{radio-over-fiber}
\newacronym{ros}{ROS}{robot operating system}
\newacronym{rpi}{RPi}{Raspberry Pi}
\newacronym{rrc}{RRC}{Radio Resource Connection}
\newacronym{rreq}{RREQ}{route request packet}
\newacronym{rsrp}{RSRP}{Reference Signals Received Power}
\newacronym{rss}{RSS}{received signal strength}
\newacronym{rssi}{RSSI}{received signal strength indicator}
\newacronym{rtc}{RTC}{real time clock}
\newacronym{rtk}{RTK}{real time kinematics}
\newacronym{rts}{RTS}{ray tracing simulator}
\newacronym{rw}{RW}{RadioWeaves}
\newacronym{rx}{RX}{receiver}
\newacronym{rzf}{RZF}{regularized zero forcing}
\newacronym{sa}{SA}{synchronization anchor}
\newacronym{scfdma}{SCFDMA}{single-carrier frequency division multiple access}
\newacronym{sdg}{SDG}{Sustainable Development Goal}
\newacronym{sdm}{SDM}{sigma-delta modulator}
\newacronym{sdn}{SDN}{software-defined network}
\newacronym{sdof}{SDoF}{sigma-delta over fiber}
\newacronym{sdr}{SDR}{software-defined radio}
\newacronym{sf}{SF}{spreading factor}
\newacronym{sfn}{SFN}{single frequency network}
\newacronym{sfp}{SFP}{small form-factor pluggable}
\newacronym{sinr}{SINR}{signal-to-interference-plus-noise ratio}
\newacronym{siso}{SISO}{single-input single-output}
\newacronym{slam}{SLAM}{simultaneous localization and mapping}
\newacronym{slc}{SLC}{spatial leakage suppression}
\newacronym{smc}{SMC}{specular multipath component}
\newacronym{smps}{SMPS}{switched mode power supply}
\newacronym{sndr}{SNDR}{signal-to-noise-and-distortion ratio}
\newacronym{snidr}{SNIDR}{signal-to-noise-and-interference-and-distortion ratio}
\newacronym{snr}{SNR}{signal-to-noise ratio}
\newacronym{soc}{SoC}{state of charge}
\newacronym{ssb}{SSB}{synchronisation signal block}
\newacronym{ssd}{SSD}{solid state drive}
\newacronym{steam}{STEAM}{science, technology, engineering, the arts, and mathematics}
\newacronym{svd}{SVD}{singular value decomposition}
\newacronym{tau}{TAU}{tracking area update}
\newacronym{tcer}{TCER}{transported to consumed energy ratio}
\newacronym{tdd}{TDD}{time division duplexing}
\newacronym{tdoa}{TDOA}{time-difference-of-arrival}
\newacronym{toa}{TOA}{time-of-arrival}
\newacronym{tof}{ToF}{time-of-flight}
\newacronym{tosm}{TOSM}{through-open-short-match}
\newacronym{trl}{TRL}{technology readyness level}
\newacronym{trp}{TRP}{Transmission Reception Point}
\newacronym{tsn}{TSN}{time-sensitive networking}
\newacronym{ttm}{TTM}{time to market}
\newacronym{ttn}{TTN}{The Things Network}
\newacronym{tx}{TX}{transmitter}
\newacronym{uav}{UAV}{unmanned aerial vehicle}
\newacronym{udp}{UDP}{User Datagram Protocol}
\newacronym{ue}{UE}{user equipment}
\newacronym{ugv}{UGV}{unmanned ground vehicle}
\newacronym{uhd}{UHD}{USRP hardware driver}
\newacronym{uhf}{UHF}{ultra-high frequency}
\newacronym{ul}{UL}{uplink}
\newacronym{ula}{ULA}{uniform linear array}
\newacronym{ummtc}{umMTC}{ultra massive machine type communication}
\newacronym{upa}{UPA}{uniform planar array}
\newacronym{ura}{URA}{uniform rectangular array}
\newacronym{urllc}{URLLC}{ultra-reliable low-latency communications}
\newacronym{usrp}{USRP}{universal software radio peripheral}
\newacronym{uv}{UV}{unmanned vehicle}
\newacronym{uwb}{UWB}{ultrawideband}
\newacronym{vep}{VEP}{virtual edge platform}
\newacronym{vlc}{VLC}{visible light communication}
\newacronym{vlp}{VLP}{visible light positioning}
\newacronym{vna}{VNA}{vector network analyzer}
\newacronym{vr}{VR}{virtual reality}
\newacronym{v2v}{V2V}{vehicle-to-vehicle}
\newacronym{wb}{WB}{wideband}
\newacronym{wpt}{WPT}{wireless power transfer}
\newacronym{wr}{WR}{White Rabbit}
\newacronym{wrsn}{WRSN}{wireless rechargeable sensor network}
\newacronym{wsn}{WSN}{wireless sensor network}
\newacronym{xr}{XR}{extended reality}
\newacronym{z3ro}{Z3RO}{zero third-order distortion}
\newacronym{zf}{ZF}{zero-forcing}
\newacronym{zmcscg}{ZMCSCG}{zero mean circularly symmetric complex Gaussian}

\newacronym{less}{LESS}{Low Energy Scheduler Solution}

\newacronym{wban}{WBAN}{wireless body area network}
\newacronym{un}{UN}{United Nations}
\newacronym{iis}{IIS}{integrated information system}
\glsdisablehyper
\usepackage{xcolor}
\usepackage{hyperref}
\usepackage{pgfplots}
\pgfplotsset{compat=1.18}
\usepgfplotslibrary{fillbetween}
\usepackage{algpseudocode}
\usepackage{balance}

\usepackage{enumitem}
\setlist[enumerate]{itemsep=0mm}

\usetikzlibrary{external} 
\newcommand{\vect}[1]{\boldsymbol{\mathrm{#1}}}
\newcommand{\mat}[1]{\boldsymbol{\mathrm{#1}}}
\newcommand{\tp}[1]{#1^{\mathrm{T}}}
\newcommand{\hr}[1]{#1^{\mathrm{H}}}
\newcommand{\inv}[1]{\left(#1\right)^{-1}}
\newcommand{\invsq}[1]{\left(#1\right)^{-2}}

\newcommand{\tr}[1]{\mathrm{tr}\,#1}
\newcommand{\diag}[1]{\mathrm{diag}\left(#1\right)}
\newcommand{\blkdiag}[1]{\mathrm{blkdiag}\left(#1\right)}
\newcommand*{\inC}[1]{\in\mathbb{C}^{#1}}
\newcommand*{\inR}[1]{\in\mathbb{R}^{#1}}

\newcommand{\abs}[1]{\left\lvert#1\right\rvert}
\newcommand{\expt}[1]{\mathbb{E}\left\{#1\right\}}
\newcommand{\cn}[2]{\ensuremath{\mathcal{C}\mathcal{N}\left(#1,#2\right)}}

\newcommand{\as}{\xrightarrow{\mathrm{a.s.}}}

\usepackage{amsthm}
\newtheoremstyle{mystyle}
  {}
  {}
  {\itshape}
  {}
  {\bfseries}
  {.}
  { }
  {}
\theoremstyle{mystyle}
\newtheorem{theorem}{Theorem}

\newtheorem{proposition}{Proposition}

\title{Energy-Saving Cell-Free Massive MIMO Precoders \\with a per-AP Wideband Kronecker Channel Model}
\name{Emanuele Peschiera$^\star$ \qquad Xavier Mestre$^\dagger$ \qquad Fran\c{c}ois Rottenberg$^\star$
\thanks{This work has been partially funded by a short-term scientific mission (STSM) grant from the COST ACTION CA20120 INTERACT.}}
\address{$^\star$ Katholieke Universiteit Leuven, Ghent, Belgium\\
$^\dagger$ Centre Tecnològic de Telecomunicacions de Catalunya, Barcelona, Spain}

\begin{document}
\ninept
\maketitle
\begin{abstract}
We study cell-free massive multiple-input multiple-output precoders that minimize the power consumed by the power amplifiers subject to per-user per-subcarrier rate constraints. The power at each antenna is generally retrieved by solving a fixed-point equation that depends on the instantaneous channel coefficients. Using random matrix theory, we retrieve each antenna power as the solution to a fixed-point equation that depends only on the second-order statistics of the channel. Numerical simulations prove the accuracy of our asymptotic approximation and show how a subset of access points should be turned off to save power consumption, while all the antennas of the active access points are utilized with uniform power across them. This mechanism allows to save consumed power up to a factor of $9\times$ in low-load scenarios.
\end{abstract}
\section{Introduction}

Energy-saving mechanisms are becoming increasingly present in wireless systems~\cite{Zhang20,Islam23,Wesemann23}. In the space domain, energy-saving strategies dynamically activate hardware resources, such as antennas and their \gls{rf} chains, or even entire \glspl{ap}, as a function of the network load. Studies on cellular massive \gls{mimo} reveal that the traffic-aware (de)activation of hardware resources is a direct consequence of considering the consumed power instead of the transmit power, as conventionally done in the precoding design~\cite{Cheng19,Senel19,Peschiera23}. From the perspective of the \glspl{pa} consumption, in a single-carrier system it is more convenient to saturate a subset of \glspl{pa}, which can work at maximal efficiency, and switch off the remaining ones~\cite{Cheng19}. Considering the \gls{bs} consumption, significant savings can be achieved by turning off some antennas and their associated \gls{rf} chains~\cite{Islam23}.

In this work, we study the optimal precoding and power allocation, in terms of \glspl{pa} consumption, for a \gls{cfmmimo} network, where several distributed \glspl{ap} jointly serve the users~\cite{Nayebi15}. This type of MIMO architecture offers substantial advantages towards \gls{6g} systems thanks to its high network coverage and macro-diversity gain~\cite{Wang23}. Studies in the literature have addressed the problem of switching on/off \glspl{ap} to minimize consumption~\cite{VanChien20,He23,Kooshki23}, but in most cases, they rely on conventional precoders and optimize the power allocations according to the power consumption. Moreover, the impact of a multicarrier transmission is often overlooked. In this work, we formulate the precoding problem as a minimization of the power consumed by the \glspl{pa} under per-subcarrier \gls{zf} constraints. In the case of centralized processing, the problem is equivalent to the cellular one in~\cite{Peschiera23}, with the crucial difference being the channel model, which incorporates the characteristics of a \gls{cfmmimo} scenario. The spatial distribution of the \glspl{ap} implies a channel having a different covariance matrix for each \gls{ap}.

Despite its advantages, the large dimension of a complete \gls{cfmmimo} observation can easily lead to significant growth in the power consumed by (i) the fronthaul links, (ii) the baseband processing at the \gls{cu}, and (iii) the \gls{rf} circuitry and power supply of all the \glspl{ap}. For this reason, leveraging \gls{rmt} and the structure of the precoding solution, we derive simplified solutions to the addressed problem. The \gls{cu} determines which \glspl{ap} to activate prior to computing the precoding matrices, and the set of active \glspl{ap} changes only with the second-order channel statistics. It can be shown that the proposed solution approaches the optimal one when the number of subcarriers increases (referred to as wideband regime since we consider a fixed subcarrier spacing).

\section{System Model}

\subsection{Signal Model}
We consider the downlink of an \gls{ofdm} \gls{cfmmimo} system with $L$ \glspl{ap}, where \gls{ap} $l$ is equipped with $M_l$ antennas, $K$ single-antenna users and $Q$ subcarriers. We consider fully-centralized processing, where all \glspl{ap} are connected to a \gls{cu}. Moreover, all the active \glspl{ap} serve all the users. The total number of antennas in the system is $\sum_{l=0}^{L-1}M_l = N$. We group the signals relative to an \gls{ap} in a contiguous manner. After precoding, the transmitted signal by all \glspl{ap} at subcarrier $q$ is
\begin{equation}
    \vect{x}_q = \tp{\left[\tp{\vect{x}_{0,q}},\dotsc,\tp{\vect{x}_{L-1,q}}\right]} = \mat{W}_q\vect{s}_q\inC{N\times 1}
\end{equation}
where $\vect{x}_{l,q}\inC{M_l\times 1}$ is the transmitted signal by \gls{ap} $l$ at subcarrier $q$, $\vect{s}_q\inC{K\times 1}$ is the symbol vector at subcarrier $q$, and $\mat{W}_q\inC{N\times K}$ is the aggregated precoding matrix corresponding to all the \glspl{ap} at subcarrier $q$. We assume that $\expt{\vect{s}_q\hr{\vect{s}}_q}=\mat{I}_K \, \forall q$. The received signal vector at subcarrier $q$ is
\begin{equation}
    \vect{r}_q = \mat{H}_q\vect{x}_q+\vect{\nu}_q\inC{K\times 1}
\end{equation}
where $\mat{H}_q=\left[\mat{H}_{0,q},\dotsc,\mat{H}_{L-1,q}\right]\inC{K\times N}$ is the aggregated channel matrix at subcarrier $q$, with $\mat{H}_{l,q}\inC{K\times M_l}$ being the channel matrix of \gls{ap} $l$ at subcarrier $q$, and $\vect{\nu}_q\sim\cn{\vect{0}_K}{\sigma_{\nu}^2\mat{I}_K}$ is the additive noise.

\subsection{Channel Model}
We model the channel of the $l$th \gls{ap} as 
\begin{equation}
\label{eq:channel}
    \mat{H}_{l,q} = \mat{D}_{\vect{\beta},l}^{\frac{1}{2}}\mat{G}_{l,q}\mat{C}_{\mathrm{AP},l}^{\frac{1}{2}}
\end{equation}
where $\mat{D}_{\vect{\beta},l}=\diag{\beta_{0,l},\dotsc,\beta_{K-1,l}}$ is a diagonal matrix collecting the \gls{lsf} coefficients for AP $l$, $\beta_{k,l}>0$ is the \gls{lsf} coefficient between user $k$ and \gls{ap} $l$, $\mat{G}_{l,q}$ models the small-scale fading and has \gls{iid} entries drawn from a $\cn{0}{1}$ distribution, and $\mat{C}_{\mathrm{AP},l}\inC{M_l\times M_l}$ expresses the correlation among the antennas of AP $l$. The above is equivalent to considering a per-\gls{ap} Kronecker channel model, while $\mat{H}_q$ follows a Weichselberger channel model~\cite{Weichselberger06}.
\begin{figure}[!ht]
\centering
\includegraphics[width=.85\linewidth]{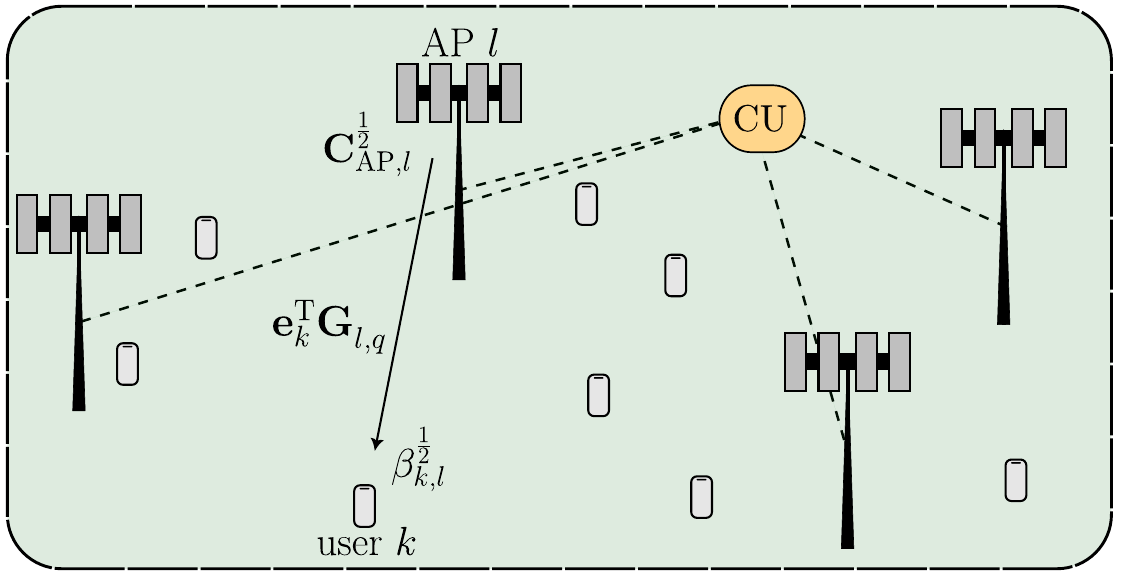}
\caption{Considered \gls{cfmmimo} scenario.}
\end{figure}

\subsection{Power Consumption Model}
Given that we consider the transmitted symbols to be of unit variance and uncorrelated among subcarriers and users, the transmit power at antenna $n$ is given by
$p_n = \sum_{q=0}^{Q-1}\tp{\vect{e}_n}\mat{W}_q\hr{\mat{W}_q}\vect{e}_n$,
where $\vect{e}_n$ is the $n$th canonical basis vector of $\mathbb{R}^N$, thus the total transmit power is
\begin{equation}
    P_\mathrm{tx} = \sum_{q=0}^{Q-1}\tr{\mat{W}_q\hr{\mat{W}_q}}
\end{equation}
where $\tr{\cdot}$ is the trace operator. Under the \gls{pa} consumption model in~\cite{Cheng19,Peschiera23}, the total power consumed by the \glspl{pa} is
\begin{equation}
\label{eq:pPAs}
    P_\mathrm{PAs} = \frac{p_\mathrm{max}^{\frac{1}{2}}}{\eta_\mathrm{max}}\sum_{n=0}^{N-1}\left(\sum_{q=0}^{Q-1}\tp{\vect{e}_n}\mat{W}_q\hr{\mat{W}_q}\vect{e}_n\right)^{\frac{1}{2}}
\end{equation}
where $p_\mathrm{max}$ and $\eta_\mathrm{max}$ are the maximal \gls{pa} transmit power and efficiency, respectively. The solutions that optimize~(\ref{eq:pPAs}) typically activate only part of the $L$ \glspl{ap}. Let us therefore identify the set of the active \glspl{ap} of the optimum solution as $\mathcal{L}_\mathrm{a}\subseteq\{0,\dotsc,L-1\}$, where the number of active \glspl{ap} is $\left|\mathcal{L}_\mathrm{a}\right|=L_\mathrm{a}$. To quantify the consumed power at the network level, we use the following model
\begin{equation}
    P_\mathrm{net} = P_\mathrm{PAs}+\sum_{l\in\mathcal{L}_\mathrm{a}}\left(P_\mathrm{fix}+P_\mathrm{c}M_{\mathrm{a},l}\right)
\end{equation}
where $P_\mathrm{fix}$ is the fixed consumption associated with the activation of an entire \gls{ap}, $P_\mathrm{c}$ is the cost (in terms of circuit power consumption) to activate a single antenna, and $M_{\mathrm{a},l}$ is the number of active antennas at \gls{ap} $l$.

\section{Precoder Design}

As a benchmark, we consider the conventional \gls{zf} precoder, which solves
\begin{equation}
\label{}
(\mathcal{P}_1):\quad
     \min_{\{\mat{W}_q\}}\;\;       P_\mathrm{tx}\quad
     \mathrm{s.t.}\;\; 		       \mat{H}_q\mat{W}_q = \tilde{\mat{D}}_{\vect{\gamma}}^{\frac{1}{2}}\sigma_\nu \;\; \forall q
\end{equation}
where $\tilde{\mat{D}}_{\vect{\gamma}}=\frac{1}{Q}\diag{\gamma_0,\dotsc,\gamma_{K-1}}$, $\gamma_k$ being the target \gls{snr} of user $k$, and $\sigma_\nu^2$ is the noise power. The \gls{zf} constraints can be easily translated into rate constraints, e.g., an achievable rate for user $k$ at subcarrier $q$ is $\log_2(1+\frac{\gamma_k}{Q})$ bits per channel use. The solution to~$(\mathcal{P}_1)$ is the per-subcarrier \gls{zf}
\begin{equation}
\label{eq:W_con}
\mat{W}_q = \hr{\mat{H}_q}\inv{\mat{H}_q\hr{\mat{H}}_q}\tilde{\mat{D}}_{\vect{\gamma}}^{\frac{1}{2}}\sigma_\nu.
\end{equation}
We are instead interested in the precoder that minimizes the power consumed by the \glspl{pa} under a \gls{zf} constraint, corresponding to
\begin{equation}
\label{}
(\mathcal{P}_2):\quad
     \min_{\{\mat{W}_q\}}\;\;  	 	P_\mathrm{PAs}\quad
     \mathrm{s.t.}\;\; 		     	\mat{H}_q\mat{W}_q = \tilde{\mat{D}}_{\vect{\gamma}}^{\frac{1}{2}}\sigma_\nu \;\; \forall q.
\end{equation}
Let us define $\tilde{\mat{H}}_q = \sigma_\nu^{-1}\tilde{\mat{D}}_{\vect{\gamma}}^{-\frac{1}{2}}\mat{H}_q$ as the normalized channel matrix.
\begin{proposition}
The solution to~$(\mathcal{P}_2)$ has the following form~\cite{Peschiera23}
\begin{equation}
\label{eq:W_opt}
    \mat{W}_q = \mat{D}_{\vect{p}}^{\frac{1}{2}}\hr{\mat{H}_q}\inv{\mat{H}_q\mat{D}_{\vect{p}}^{\frac{1}{2}}\hr{\mat{H}_q}}\tilde{\mat{D}}_{\vect{\gamma}}^{\frac{1}{2}}\sigma_\nu
\end{equation}
where $\mat{D}_{\vect{p}} = \diag{p_0,\dotsc,p_{N-1}}$ contains the powers transmitted by the antennas. The transmit power at antenna $n$ can be written as the solution to the set of equations (for $n=1,\dotsc,N$)
\begin{equation}
\label{eq:p_n}
    p_n = \sum_{q=0}^{Q-1}\tp{\vect{e}_n}\mat{D}_{\vect{p}}^{\frac{1}{2}}\hr{\tilde{\mat{H}}_q}\invsq{\tilde{\mat{H}}_q\mat{D}_{\vect{p}}^{\frac{1}{2}}\hr{\tilde{\mat{H}}_q}}\tilde{\mat{H}}_q\mat{D}_{\vect{p}}^{\frac{1}{2}}\vect{e}_n
\end{equation}
which can be retrieved via iterative fixed-point algorithm. 
\end{proposition}
\noindent The structure of the precoder in~(\ref{eq:W_opt}) allows one to allocate more power to some antennas, which can work closer to the \gls{pa} saturation level. Turning off some antennas/\glspl{ap} implies further savings in $P_\mathrm{net}$. Note that the \gls{cu} needs to (i) compute the antenna powers in~(\ref{eq:p_n}) every time the small-scale fading changes (i.e., every channel coherence time) and (ii) know the instantaneous channel coefficients of all the \glspl{ap}, which corresponds to $NKQ$ elements. These operations impose challenging requirements to the fronthaul capacity of a \gls{cfmmimo} system due to the large number of antennas/\glspl{ap}. For this reason we propose, in the next section, a \gls{rmt} approach that reduces the complexity of retrieving the \glspl{ap} powers.

\section{Asymptotic Analysis}

In this section, we consider the following asymptotic regime.

$\mathbf{(As1)}$: $M_l\to\infty\;\forall l\in\{0,\dotsc,L-1\}$ and $K\to\infty$. Moreover, defining $c_l=\frac{K}{M_l}$, there exist $b>a>0$ for which $a<\min_l\liminf_Kc_l<\max_l\limsup_Kc_l<b$.

Thanks to the hardening effect that happens in this regime, the random quantities in~(\ref{eq:p_n}) behave as deterministic ones~\cite{RMTbook}. Nonetheless, the solution to the system of equations in~(\ref{eq:p_n}) has  $N$ components, therefore $\mathbf{(As1)}$ implies that the dimension of the solution grows infinitely large. The number of \glspl{ap} $L$ is instead fixed. This motivates the introduction of the following assumption.

$\mathbf{(As2)}$: $\mat{D}_{\vect{p}}=\blkdiag{\left\{p_l\mat{I}_{M_l}\right\}_{l=0}^{L-1}}$, where $\blkdiag{\cdot}$ is the block-diagonal operator, corresponding to a uniform power allocation per \gls{ap}. Therefore, the effective dimension of the solution is fixed and equal to $L$.

The above assumption is an approximation, indeed one can retrieve the exact solution to~(\ref{eq:p_n}) and observe that more power is allocated to the least correlated antennas. However, numerical simulations will demonstrate that the results obtained under $\mathbf{(As2)}$ have a negligible performance loss compared to the optimal ones. The following assumptions are also made.

$\mathbf{(As3)}$: $\forall l\in\{0,\dotsc,L-1\}$, $\mat{G}_{l,q}\inC{K\times M_l}$ has \gls{iid}~entries with zero mean and unit variance.

$\mathbf{(As4)}$: $\forall l\in\{0,\dotsc,L-1\}$, $\mat{D}_{l}\inR{K\times K}$, $\mat{C}_{\mathrm{AP},l}\inC{M_l\times M_l}$ and $\mat{D}_{\vect{p}}\inR{N\times N}$ are deterministic matrices with bounded spectral norms. Moreover, $\mat{D}_{\vect{p}}$ has at least $K+1$ non-zero entries.

While in the previous section we did not rely on any specific channel model, we now consider the per-\gls{ap} Kronecker channel in~(\ref{eq:channel}). By defining $\mat{D}_l=\sigma_\nu^{-2}\tilde{\mat{D}}_{\vect{\gamma}}^{-1}\mat{D}_{\vect{\beta},l}$, we can write
\begin{equation}
\label{eq:block-Kronecker}
    \tilde{\mat{H}}_q\mat{D}_{\vect{p}}^{\frac{1}{2}}\hr{\tilde{\mat{H}}_q} = \sum_{l=0}^{L-1}p_l^\frac{1}{4}\mat{D}_{l}^{\frac{1}{2}}\mat{G}_{q,l}\mat{C}_{\mathrm{AP},l}\hr{\mat{G}_{q,l}}\mat{D}_{l}^{\frac{1}{2}}p_l^\frac{1}{4}.
\end{equation}
The above model is analyzed by the authors in~\cite{Couillet11}. In our case, we want to characterize $\tp{\vect{e}_n}\mat{D}_{\vect{p}}^{\frac{1}{2}}\hr{\tilde{\mat{H}}_q}\invsq{\tilde{\mat{H}}_q\mat{D}_{\vect{p}}^{\frac{1}{2}}\hr{\tilde{\mat{H}}_q}}\tilde{\mat{H}}_q\mat{D}_{\vect{p}}^{\frac{1}{2}}\vect{e}_n$, which corresponds to the left-hand side of~(\ref{eq:p_n}). It turns out that this quantity is asymptotically equivalent to a deterministic one (almost surely) as the system dimensions increase to infinity.
Let us now introduce the variables $b_l$, $l\in\{0,\dotsc,L-1\}$, which form the unique non-negative solutions to the $L$ equations
\begin{equation}
    b_l = \frac{1}{K}\tr{p_l^\frac{1}{2}\mat{D}_{l}\left(\sum_{l'=0}^{L-1}\frac{1}{M_{l'}}\sum_{m=0}^{M_{l'}-1}\frac{\xi_{l',m}}{1+c_{l'}\xi_{l',m}b_{l'}}p_{l'}^\frac{1}{2}\mat{D}_{l'}\right)^{-1}}
\end{equation}
where $\xi_{l,m}$ is the $m$th eigenvalue of $\mat{C}_{\mathrm{AP},l}$.
The variables $\dot{b}_l$, $l\in\{0,\dotsc,L-1\}$, are defined as
\begin{equation}
\begin{bmatrix}
\dot{b}_0\\
\vdots\\
\dot{b}_{L-1}
\end{bmatrix}
=
\inv{
\mat{I}_L-
\mat{A}
}
\begin{bmatrix}
\frac{1}{K}\tr{p_0^\frac{1}{2}\mat{D}_{0}\mat{B}^{-2}}\\
\vdots\\
\frac{1}{K}\tr{p_{L-1}^\frac{1}{2}\mat{D}_{L-1}\mat{B}^{-2}}
\end{bmatrix}
\end{equation}
with $\mat{A}\inR{L\times L}$, $\left[\mat{A}\right]_{l,l'} = \frac{1}{K}p_{l}^\frac{1}{2}\tr{\mat{D}_{l}\mat{B}^{-2}p_{l'}^\frac{1}{2}\mat{D}_{l'}^\prime}$,
$\mat{B} = \sum_{l=0}^{L-1}\frac{1}{M_l}\\\sum_{m=0}^{M_l-1}\frac{\xi_{l,m}}{1+c_l\xi_{l,m}b_l}p_l^\frac{1}{2}\mat{D}_{l}$, and
$\mat{D}_{l}' = \frac{1}{M_l}\sum_{m=0}^{M_l-1}\frac{c_l\xi_{l,m}^2}{\left(1+c_l\xi_{l,m}b_l\right)^2}\mat{D}_{l}$. 
We now have all the ingredients to express the deterministic equivalent of the individual transmit powers at the \glspl{ap}.
\begin{theorem}
\label{th1}
    Under $(\mathbf{As1})-(\mathbf{As4})$, for a fixed $L$ and considering the channel model in~(\ref{eq:channel}), we have 
    $\left(p_l - \overline{p_l}\right) \as 0$, where 
    \begin{equation}
    \label{eq:pbar}
    \overline{p_l} = \frac{1}{M_l^2}\tr{\mat{D}_{\mathrm{sel},l}\mat{D}_{\vect{p}}^{\frac{1}{2}}\mat{V}}
    \end{equation}
    \begin{equation}
        \mat{V} = \blkdiag{\left\{\invsq{\mat{I}_{M_l}+c_lb_l\mat{C}_{\mathrm{AP},l}}c_l\dot{b}_l\mat{C}_{\mathrm{AP},l}\right\}_{l=0}^{L-1}}
    \end{equation}
    and $\mat{D}_{\mathrm{sel},l} = \blkdiag{\left\{\delta_{l,l'}\mat{I}_{M_{l'}}\right\}_{l'=0}^{L-1}}$, where $\delta_{l,l'}$ is the Kronecker delta function, selects the antenna indexes of the \gls{ap} $l$.
\end{theorem}
\begin{proof}
    Based on~\cite[Th.~1]{Couillet11}, it is omitted due to lack of space.
\end{proof}

Eq.~(\ref{eq:pbar}) is a fixed-point equation in the \glspl{ap} powers and can be solved by using standard fixed-point methods (e.g.,~\cite[Algorithm 1]{Peschiera23}). The variables $\{b_l\}$ can also be computed via fixed-point algorithm, while $\{\dot{b}_l\}$ are computed in closed form as a function of $\{b_l\}$. Note that the computation of $\left\{\overline{p_l}\right\}$ requires knowing only the second-order channel statistics (i.e., \gls{lsf} and correlation coefficients), which are much more stable over time.
The \gls{cu}, in our proposed strategy, performs the following steps:
\begin{enumerate}
    \item Receive the updated \gls{lsf} coefficients from all the $L$ \glspl{ap}. 
    \item Compute $\{\overline{p_l}\}$ by solving~(\ref{eq:pbar}) $\forall l\in\{0,\dotsc,L-1\}$.
    \item Select the active \glspl{ap}, $\mathcal{L}_\mathrm{a}^\star=\{l\in\{0,\dotsc,L-1\} \mid \overline{p_l}>\epsilon\}$.\footnote{The value of $\epsilon$ is related to the performance of the fixed-point algorithm and can be set empirically.}
    \item For the whole time interval where the \gls{lsf} remains constant, receive from the \glspl{ap} in $\mathcal{L}_\mathrm{a}^\star$ the updated instantaneous channel coefficients, whose number is reduced to $\sum_{l\in\mathcal{L}_\mathrm{a}^\star}M_lKQ$, and compute the precoding matrices via~(\ref{eq:W_opt}). Then, send the precoding coefficients to the \glspl{ap} in $\mathcal{L}_\mathrm{a}^\star$.
\end{enumerate}

\section{Performance Evaluation}

\begin{figure*}[!t]
\centering
\includegraphics{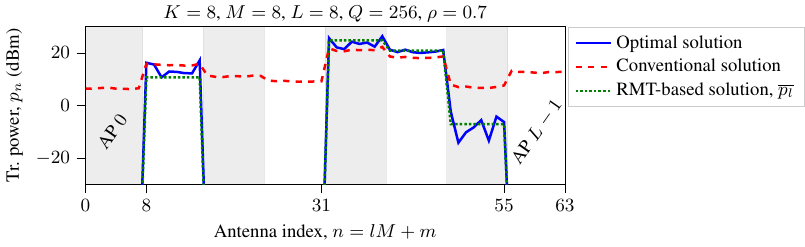}
\caption{Transmit power as a function of the antenna index for a single realization of a per-AP Kronecker channel. The graph shows the transmit power per antenna for: the optimal solution~(\ref{eq:W_opt}), the conventional one~(\ref{eq:W_con}), and the one based on \gls{rmt} (Theorem~\ref{th1}).}
\label{fig:2}
\end{figure*}

In the numerical simulations, we assume a square deployment area of $1\times1$ km\textsuperscript{2} with $16$ possible \gls{ap} locations on a regular grid, where the positions of the $L$ \glspl{ap} are randomly selected at every channel realization. The users are uniformly distributed within the scenario, with a minimum distance of $10$ m to an \gls{ap}, and their locations are randomly generated at every channel realization. We consider a constant number of antennas per \gls{ap}, set to $M=8$. When characterizing the average gains in power consumption, we consider a number of users varying from $K=1$ to $K=19$, as well as a number of \glspl{ap} ranging from $L=2$ to $L=10$. 
The target \glspl{snr} $\{\gamma_k\}$ are randomly generated between $1$ dB and $20$ dB with a uniform distribution, while the noise power is $10\log_{10}\left(\sigma_\nu^2\right)+30 = -96$ dBm.
The \gls{lsf} coefficients are computed (in dB) as $10\log_{10}(\beta_{k,l}) = -30.5 - 37.6\log_{10}d_{k,l}+F_{k,l}$~\cite{3gpp},
where $d_{k,l}$ is the distance (in meters) between the user $k$ and the \gls{ap} $l$, and $F_{k,l}\sim\mathcal{N}(0,16)$ represents the shadow fading. 
The correlation between the antennas of an \gls{ap} follows the exponential model, i.e., $\left[\mat{C}_{\mathrm{AP},l}\right]_{m,m'}=\rho^{\abs{m-m'}}$~\cite{Loyka01}, where we set $\rho=0.7$. We consider a \gls{pa} maximal transmit power and efficiency of $p_\mathrm{max} = 3$ W and $\eta_\mathrm{max} = 0.34$, respectively~\cite{Qorvo}. The fixed power consumption is set to $P_{\mathrm{fix}} = 15$ W, while the circuit power consumption per antenna is set to $P_\mathrm{c} = 0.7$ W~\cite{Auer11}.

\subsection{Single Channel Realization}
Fig.~\ref{fig:2} shows the antenna powers for a single channel realization, $K=8$ users, $L=8$ \glspl{ap} and $Q=256$ subcarriers. We observe that the powers among the antennas of an \gls{ap} share similar values, with a certain variability that is likely due to the correlation among antennas. The first and last antennas, being the least correlated given that we are considering a uniform linear array, are allocated more power. The \gls{rmt}-based solution uses a uniform power per \gls{ap}, and this power approximates the average of the optimal solution among the $M$ antennas of the \gls{ap}. Fig.~\ref{fig:2} indicates also that the results in Theorem~\ref{th1}, even though derived for $M,K\to\infty$, are accurate for finite and relatively low values of $M$ and $K$.

It is interesting to notice how the solution that minimizes the \glspl{pa} consumed power turns on/off entire \glspl{ap}. All the antennas of the active \glspl{ap} are utilized, and vice versa. Instead, the solution that optimizes transmit power always activates all the \glspl{ap}. The consumed power solution utilizes $L_\mathrm{a}=4$ \glspl{ap}, corresponding to a halving in the fronthaul requirements. The gain in network consumption (i.e., the ratio between $P_\mathrm{net}$ by using the proposed solution and $P_\mathrm{net}$ by using the conventional solution), equals $1.78$ both using the optimal solution and the \gls{rmt}-based one, while the gains in \glspl{pa} consumption are $1.28$ and $1.27$, respectively.

\subsection{Impact of Number of Subcarriers}
We observe, via numerical experiments, that the uniform power allocation per AP in $\mathbf{(As2)}$ approximates well the optimal solution when $Q$ grows. This is in line with the findings in~\cite{Peschiera23} for cellular systems, where uniform allocation among all the \gls{bs} antennas is optimal (in terms of \glspl{pa} consumption) when $Q$ is large and there is no spatial correlation. Fig.~\ref{fig:4} depicts the ratio between the values of $P_\mathrm{PAs}$ by using the optimal~(\ref{eq:p_n}) and asymptotic~(\ref{eq:pbar}) transmit powers as a function of $Q$.
For a few subcarriers, the optimal solution activates only a part of the antennas of certain \glspl{ap}. When $Q$ is sufficiently large, instead, the two strategies achieve the same performance. This tells us that the essential information is given by the set $\mathcal{L}_\mathrm{a}^\star$. Once we know which \glspl{ap} to activate, we can use a uniform power allocation among them with a negligible performance loss.
\begin{figure}[!ht]
\centering
\includegraphics{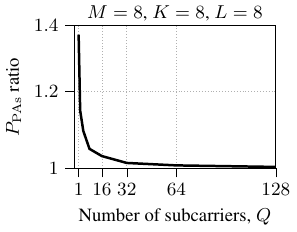}
\caption{Ratio between \gls{pa} consumption by using the asymptotic solution and the optimal one as a function of the number of subcarriers, averaged over ${10}^2$ channel realizations.}
\label{fig:4}
\end{figure}

\subsection{Achievable Gains in Network Consumption}
Fig.~\ref{fig:3} illustrates the average achievable savings in network consumption by using the asymptotic solution $\left\{\overline{p_l}\right\}$ over the conventional one for $Q=256$ subcarriers. As expected, the more the \glspl{ap}, the larger the gains, as the consumed power solution activates only a fraction of the $L$ \glspl{ap} in low and medium loads. The savings decrease as $K$ increases because (i) more users require the activation of more antennas and (ii) it is more likely that a user will be in the proximity of each \gls{ap}. For $L=10$ \glspl{ap}, the gain reaches a factor of $9\times$ when there is $K=1$ user and remains large also for $K=7$, when it equals $2$. For $L=8$ \glspl{ap} and $K=9$ users, the saving is still equal to $1.5$ (corresponding to a $50\%$ less consumed power). In general, the ratios become lower or equal than $1.5$ when $K>L$.
\begin{figure}[!ht]
\centering
\includegraphics{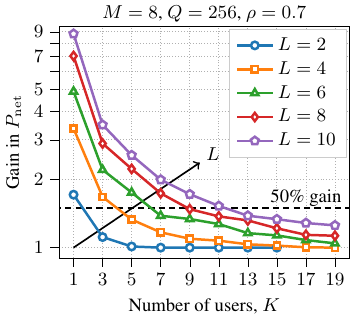}
\caption{Gain in network power consumption (of utilizing the \gls{rmt}-based solution against the conventional one), averaged over ${10}^2$ channel realizations, as a function of the number of users for different number of \glspl{ap}, $M=8$ antennas, and $Q=256$ subcarriers.}
\label{fig:3}
\end{figure}

\section{Conclusion}
The problem of designing \gls{cfmmimo} precoders that minimize the \glspl{pa} consumption subject to rate constraints with centralized processing has been addressed. Under a per-\gls{ap} Kronecker channel model, we applied \gls{rmt} to derive simple fixed-point equations in the antenna powers that depend only on the channel covariances and are accurate when the number of subcarriers grows. The optimal strategy consists of activating all the antennas of certain \glspl{ap} with uniform powers and deactivating all the antennas of the remaining \glspl{ap}. By switching off the circuits of the non-active \glspl{ap}, we achieved significant savings in the network power consumption in low and medium loads. Perspectives include the consideration of local processing (i.e., performed at the individual \glspl{ap}) and the analysis of user-centric implementations.

\vfill\pagebreak
\balance
\bibliographystyle{IEEEbib}
\bibliography{refs}

\begin{thebibliography}{10}

\bibitem{Zhang20}
S.~Zhang, S.~Xu, G.Y. Li, and E.~Ayanoglu,
\newblock ``First 20 years of green radios,''
\newblock {\em IEEE Trans. Green Commun. Netw.}, vol. 4, no. 1, pp. 1--15,
  2020.

\bibitem{Islam23}
T.~Islam, D.~Lee, and S.S. Lim,
\newblock ``Enabling network power savings in {5G}-advanced and beyond,''
\newblock {\em IEEE J. Sel. Areas Commun.}, vol. 41, no. 6, pp. 1888--1899,
  2023.

\bibitem{Wesemann23}
S.~Wesemann, J.~Du, and H.~Viswanathan,
\newblock ``Energy efficient extreme {MIMO}: Design goals and directions,''
\newblock {\em IEEE Commun. Mag.}, pp. 1--7, 2023.

\bibitem{Cheng19}
H.V. Cheng, D.~Persson, and E.G. Larsson,
\newblock ``Optimal {MIMO} precoding under a constraint on the amplifier power
  consumption,''
\newblock {\em IEEE Trans. Commun.}, vol. 67, no. 1, pp. 218--229, 2019.

\bibitem{Senel19}
K.~Senel, E.~Björnson, and E.G. Larsson,
\newblock ``Joint transmit and circuit power minimization in massive {MIMO}
  with downlink {SINR} constraints: When to turn on massive {MIMO}?,''
\newblock {\em IEEE Trans Wirel. Commun.}, vol. 18, no. 3, pp. 1834--1846,
  2019.

\bibitem{Peschiera23}
E.~Peschiera and F.~Rottenberg,
\newblock ``Energy-saving precoder design for narrowband and wideband massive
  {MIMO},''
\newblock {\em IEEE Trans. Green Commun. Netw.}, 2023.

\bibitem{Nayebi15}
E.~Nayebi, A.~Ashikhmin, T.L. Marzetta, and H.~Yang,
\newblock ``Cell-free massive {MIMO} systems,''
\newblock in {\em Proc. 49th Asilomar Conf. Signals Syst. Comput.}, 2015, pp.
  695--699.

\bibitem{Wang23}
C.~Wang et~al.,
\newblock ``On the road to {6G}: Visions, requirements, key technologies, and
  testbeds,''
\newblock {\em IEEE Commun. Surveys Tuts.}, vol. 25, no. 2, pp. 905--974, 2023.

\bibitem{VanChien20}
T.~Van~Chien, E.~Bj\"{o}rnson, and E.G. Larsson,
\newblock ``Joint power allocation and load balancing optimization for
  energy-efficient cell-free massive {MIMO} networks,''
\newblock {\em IEEE Trans. Wirel. Commun.}, vol. 19, no. 10, pp. 6798--6812,
  2020.

\bibitem{He23}
Q.~He, \"{O}.T. Demir, and C.~Cavdar,
\newblock ``Dynamic {AP} selection and cluster formation with minimal switching
  for green cell-free massive {MIMO} networks,''
\newblock in {\em Proc. Joint Eur. Conf. Netw. Commun. 6G Summit}, 2023, pp.
  234--239.

\bibitem{Kooshki23}
F.~Kooshki, A.G. Armada, M.M. Mowla, A.~Flizikowski, and S.~Pietrzyk,
\newblock ``Energy-efficient sleep mode schemes for cell-less {RAN} in {5G} and
  beyond {5G} networks,''
\newblock {\em IEEE Access}, vol. 11, pp. 1432--1444, 2023.

\bibitem{Weichselberger06}
W.~Weichselberger, M.~Herdin, H.~Ozcelik, and E.~Bonek,
\newblock ``A stochastic {MIMO} channel model with joint correlation of both
  link ends,''
\newblock {\em IEEE Trans. Wireless Commun.}, vol. 5, no. 1, pp. 90--100, 2006.

\bibitem{RMTbook}
R.~Couillet and M.~Debbah,
\newblock {\em Random Matrix Methods for Wireless Communications},
\newblock Cambridge Univ. Press, New York, NY, USA, 2011.

\bibitem{Couillet11}
R.~Couillet, M.~Debbah, and J.W. Silverstein,
\newblock ``A deterministic equivalent for the analysis of correlated {MIMO}
  multiple access channels,''
\newblock {\em IEEE Trans. Inf. Theory}, vol. 57, no. 6, pp. 3493--3514, 2011.

\bibitem{3gpp}
3GPP,
\newblock {\em Further advancements for E-UTRA physical layer aspects (Release
  9)},
\newblock 3GPP TS 36.814, Mar. 2017.

\bibitem{Loyka01}
S.L. Loyka,
\newblock ``Channel capacity of {MIMO} architecture using the exponential
  correlation matrix,''
\newblock {\em IEEE Commun. Lett.}, vol. 5, no. 9, pp. 369--371, 2001.

\bibitem{Qorvo}
Qorvo QPA4501,
\newblock {\em {4.4 - 5.0 GHz, 3 Watt, 28 Volt GaN Power Amplifier Module}},
  2021,
\newblock {Rev.~C}.

\bibitem{Auer11}
G.~Auer et~al.,
\newblock ``How much energy is needed to run a wireless network?,''
\newblock {\em IEEE Wireless Commun.}, vol. 18, no. 5, pp. 40--49, 2011.

\end{thebibliography}

\end{document}